\documentclass[10pt]{article}
 \usepackage{graphicx} 

  \usepackage{amsmath,amssymb,amstext,amsthm}

  \newtheorem{theorem}{Theorem}[section]
  \newtheorem{lemma}[theorem]{Lemma}
  \newtheorem{proposition}[theorem]{Proposition}

  \newtheorem{claim}{Claim}

\newcommand{\pmc}{potential maximal clique}
\newcommand{\cO}{\mathcal{O}}
\newcommand{\sm}{\setminus}

\newtheorem{Rule}{Rule}

\title{Subexponential Parameterized Algorithm for Minimum Fill-in
\thanks{This research was partially supported by the Research Council of Norway.}
}

\author{Fedor V. Fomin
\thanks{Department of Informatics, University of Bergen,  Bergen, Norway,  \{fedor.fomin,yngve.villanger\}@ii.uib.no} \and  Yngve Villanger~$^\dagger$
}

\begin{document}

\date{}
\maketitle

\begin{abstract}
The \textsc{Minimum Fill-in} problem is to decide if a graph can be triangulated by adding at most $k$ edges.
Kaplan, Shamir, and Tarjan [FOCS 1994] have shown that the problem is solvable  in time $\cO(2^{\cO({k})} +k^2 nm)$ on graphs with $n$ vertices and $m$ edges and thus is 
fixed parameter tractable. 
Here, we give
  the first  subexponential parameterized algorithm solving  \textsc{Minimum Fill-in}   in time $\cO(2^{\cO(\sqrt{k}\log{k})} +k^2 nm)$. 
This 
substantially lower the complexity of the problem.
  Techniques developed for   \textsc{Minimum Fill-in}  can be used to obtain subexponential parameterized algorithms for several  related problems including \textsc{Minimum Chain Completion}, 
  \textsc{Chordal Graph Sandwich}, and   \textsc{Triangulating Colored Graph}. 
 \end{abstract}

\section{Introduction}
A graph is \emph{chordal} (or triangulated) if every cycle of length at least four contains a chord, i.e. an edge between nonadjacent vertices of the  cycle. The \textsc{Minimum Fill-in}  problem (also known as  \textsc{Minimum Triangulation} and  \textsc{Chordal Graph Completion}) is  

\begin{center}
\fbox{\begin{minipage}{13cm}
\noindent  \textsc{Minimum Fill-in}\\
{\sl Input:} A graph $G =(V,E) $ and a non-negative integer $k$.\\
{\sl Question:} Is there $F\subseteq [V]^{2}$, $|F|\leq k$, such that  graph $H=(V,E\cup F)$ is chordal? 
 
\end{minipage}}
\end{center}

 The name fill-in is due to  the fundamental problem arising in sparse matrix computations which was studied intensively in the past.    
 During   Gaussian eliminations of
    large sparse matrices  new non-zero elements called fill can replace original zeros thus increasing storage requirements and running time needed to solve the system. The problem of finding the right elimination ordering minimizing the amount of fill elements can be expressed as the \textsc{Minimum Fill-in} problem on graphs  \cite{Parter61,Rose72}.   See also \cite[Chapter 7]{DavisT-Direct} for a more recent overview of   related problems and techniques. 
  Besides sparse matrix computations,   applications of  \textsc{Minimum Fill-in} can be found in  
 database management 
 \cite{BeeriFMY83}, artificial intelligence, and the theory of Bayesian statistics 
 \cite{ChungM94,Geman:1990dp,Lauritzen:1990,WongWB02}. The survey of Heggernes    \cite{Heggernes06} gives an overview of techniques and applications of  minimum and minimal triangulations.

\textsc{Minimum Fill-in}  (under the name \textsc{Chordal Graph Completion}) was one of the 12 open  problems presented at the end of the first edition of Garey and Johnson's book  \cite{GareyJ79} and  it was  proved to be NP-complete by Yannakakis 
\cite{Yannakakis81}. Kaplan et al.  proved  that \textsc{Minimum Fill-in}  is  fixed parameter tractable by giving an algorithm of running time  $\cO(m16^k)$ in \cite{focs/KaplanST94}  and  improved the running time to  $\cO (k^6 16^k+k^2mn)$  in
\cite{KaplanST99}, where $m$ is the number of edges and $n$ is the number of vertices of the input graph. There were several algorithmic improvements resulting in decreasing the constant in the base of  the exponents. 
In 1996, Cai \cite{Cai96}, reduced the running time to
 $\cO((n+m) \frac{4^k}{k+1})$. The fastest parameterized algorithm known prior to our work is the recent  algorithm of Bodlaender et al.  with running time  
$\cO(2.36^k +k^2mn)$ \cite{Bodlaender:2011lr}.

 In this paper we give the first subexponential parameterized algorithm for \textsc{Minimum Fill-in}.
 The last chapter of Flum and Grohe's book 
 \cite[Chapter~16]{FlumGrohebook}
concerns  subexponential fixed parameter 
tractability,  the complexity class  SUBEPT,   which, loosely speaking---we skip here some technical conditions---is  the class of problems solvable in time   $2^{o(k)}n^{O(1)}$, where $n$ is the input length and $k$ is the parameter. 
Subexponential fixed parameter tractability is intimately linked with the theory of exact exponential algorithms for hard problems, which   are better than the trivial exhaustive search, though still exponential \cite{Fomin:2010mo}. 
  Based on the fundamental results of  Impagliazzo et al. 
 \cite{ImpagliazzoPZ01}, Flum and Grohe established that most of the natural parameterized problems are not in  SUBEPT unless Exponential Time Hypothesis (ETH) fails.   Until  recently, the only notable exceptions of problems in SUBEPT were  the  problems on planar graphs, and more generally, on graphs excluding some fixed graph as a minor
   \cite{DemaineFHT05jacm}. In 2009, Alon  et al. \cite{AlonLS09} used a novel application of color coding to show that parameterized \textsc{Feedback Arc Set in Tournaments} is in SUPEPT.   
 \textsc{Minimum Fill-in} is the first problem on general graphs which appeared to be in SUBEPT.

\medskip\noindent
  \textbf{General overview of our approach.}
Our main tool in obtaining  subexponential algorithms is the theory of minimal triangulations and potential maximal cliques of  Bouchitt\'e and Todinca  \cite{BouchitteT01}. This theory was developed in contest of computing the treewidth of special graph classes and was used later   in  exact exponential  algorithms 
\cite{FominKTV08,FominV10}.  A set of vertices $\Omega$ of a graph $G$ is a potential maximal clique if there is a minimal triangulation such that $\Omega$ is a maximal clique  in this triangulation. Let $\Pi$ be the set of all potential maximal cliques in graph $G$. The importance of potential maximal cliques is that if we are given the set $\Pi$, 
 then by using the machinery from  \cite{BouchitteT01,FominKTV08}, it is possible   compute an optimum triangulation in time up to polynomial factor proportional to $|\Pi|$.   
 
 Let $G$ be an $n$-vertex  graph  and  $k$ be the parameter. If $(G,k)$ is a YES instance of the 
   \textsc{Minimum Fill-in}  problem, then every maximal clique of every optimum triangulation   is obtained from some potential maximal clique of $G$ by adding  at most $k$ fill edges. We call such potential maximal clique   \emph{vital}.  To give a general  overview of our algorithm, we start with the approach that   does not work directly, and then explain what has to be  changed to suceed. 
The algorithm consists of three main steps. 

\begin{itemize}
\item[Step A.] Apply a kernelization algorithm that in time $n^{O(1)}$  reduces the problem an instance to instance of  size polynomial in  $k$;
\item[Step B.]  Enumerate all vital potential maximal cliques of an $n$-vertex graph in time $n^{o(k/\log{k})}$. By Step~A, $n=k^{O(1)}$, and thus the running time of enumeration algorithm and the number of vital potential maximal cliques is 
$2^{o(k)}$;
 \item[Step C.]  Apply the theory of potential maximal clique to solve  the  problem in time proportional to the number of vital potential maximal cliques, which is $2^{o(k)}$.
 \end{itemize}
 
Step~A, kernelization for  \textsc{Minimum Fill-in}, was known prior to our work. In 1994, Kaplan et al. gave a kernel with  $\cO(k^5)$ vertices. Later the kernelization was improved  to $O(k^3)$  in\cite{KaplanST99} and then to $2k^2 + 4k$ in \cite{NatanzonSS00}. Step~C, with some modifications,   is similar to the algorithm from  \cite{BouchitteT01,FominKTV08}. This is  Step~B which does not work  and instead of enumerating vital potential  maximal cliques we make a ``detour". We use   branching (recursive) algorithm that in subexponential time outputs subexponential  number of graphs   avoiding a specific combinatorial structure, non-reducible graphs. In non-reducible graphs we are able to enumerate vital potential  maximal cliques.
Thus Step~B is replaced with 
 \begin{itemize}
  \item[Step~B1.] 
Apply  branching algorithm to 
generate $n^{\cO(\sqrt{k})}$ non-reducible instances such that  the original instance is a  YES instance if and only if at least one of the generated non-reducible instances is a YES instance;
\item[Step~B2.] 
Show that if   $G$ is non-reducible,  then  all vital potential maximal cliques of $G$ can be enumerated in time  $n^{\cO(\sqrt{k})}$.

 \end{itemize}
  Putting together Steps~A, B1, B2, and C, we obtain the subexponential algorithm.

 \medskip
 It follows from our results that several other problems belong to SUBEPT. 
  We  show that within time $\cO(2^{\cO(\sqrt{k}\log{k})} +k^2nm)$ it is possible to solve   
 \textsc{Minimum Chain Completion},  and   \textsc{Triangulating Colored Graph}. For   
  \textsc{Chordal Graph Sandwich},  and we show that deciding if a sandwiched chordal graph $G$ can be obtained from $G_1$ by adding at most $k$ fill edges,  is possible in time   $\cO(2^{\cO(\sqrt{k}\log{k})} +k^2nm)$.
  
A \emph{chain} graph is a bipartite graph   where the sets of neighbors of vertices form an inclusion  chain. In the  \textsc{Minimum Chain Completion}  problem, we are asked if a bipartite graph can be turned into a chain graph by adding at most $k$ edges. The problem was introduced by Golumbic  \cite{Golumbic80} and Yannakakis \cite{Yannakakis81}. The concept of chain graph has surprising applications in ecology \cite{Mannila:2007,Patterson:1986bs}. 
Feder et al. in \cite{FederMT09} gave approximation algorithms for this problem.

The   {\sc Triangulating Colored Graph} problem  is a generalization of \textsc{Minimum Fill-in}.  The instance is a graph with some of its vertices colored; the task is to  add at most $k$  fill edges such that the resulting graph is chordal and no fill edge is monochromatic. We postpone the formal definition of the problem till Section~\ref{sec:other_problems}. The problem was studied intensively because of its close relation to 
\textsc{Perfect Phylogeny Problem}---fundamental and long-standing problem for numerical taxonomists
\cite{BonetPWY96,Buneman:1974ul,KannanW92}. 
The  {\sc Triangulating Colored Graph} problem is 
$NP$-complete~\cite{BodlaenderFW92} and   $W[t]$-hard for any $t$, when parametrized by the number of 
colours \cite{BodlaenderFHWW00}.  However, the problem is fixed parameter tractable when parameterized by the number of fill edges. 

 In \textsc{Chordal Graph Sandwich} we are given two graphs $G_1$ and $G_2$ on the same vertex set, and the question is if there is a chordal graph $G$ which is a supergraph of $G_1$ and a subgraph of $G_2$.  
The problem is a generalization of   {\sc Triangulating Colored Graph}. We refer to the paper of Golumbic et al.  
\cite{GolumbicKS95} for a general overview of   graph sandwich problems.

\medskip 
The remaining part of the paper is organized as follows. 
Section~\ref{sec:prelim} contains definitions and preliminary results. In Section~\ref{sec:branch}, we give branching algorithm, Step~B1.   Section~\ref{sec:pmc}  provides algorithm  enumerating  vital potential maximal cliques in non-reducible graphs, Step~B2. This is the most important part of our algorithm. It is based on a new insight into the combinatorial structure of potential maximal cliques. In Section~\ref{ss:dp}, we show how to adapt the algorithm from   \cite{BouchitteT01,FominKTV08}  to implement Step~C. The main algorithm is given in Section~\ref{sec:puttingtogether}. 
In Section~\ref{sec:other_problems}, we show how the     ideas used for  \textsc{Minimum Fill-in} can be used to obtain subexponential algorithms for other problems. To implement our strategy for   \textsc{Chordal Graph Sandwich}, we   have to provide a polynomial kernel for this problem. We conclude with   open problems in Section~\ref{sec:conclusion}.

\section{Preliminaries}\label{sec:prelim}
We denote by $G=(V,E)$ a finite, undirected and simple graph
with vertex set $V(G)=V$ and edges set $E(G)=E$. We also use   $n$ to denote the number of vertices and $m$ the number of edges in $G$.
For a nonempty subset $W \subseteq V$, the subgraph of $G$ induced
by $W$ is denoted by $G[W]$. 
We say that a vertex set $W\subseteq V$ is \emph{connected} if $G[W]$ is connected.
The \emph{open neighborhood} of a vertex $v$ is $N(v)=\{u\in V:~uv \in E\}$ and the 
\emph{closed neighborhood} is $N[v] = N(v) \cup \{v\}$. 
For a vertex set $W\subseteq V$ we put  $N(W) = \bigcup_{v \in W} N(v)\sm W$ and $N[W] = N(W) \cup W$.
Also for  $W \subset V$  we  define $\textbf{fill}_G(W)$, or simple $\textbf{fill}(W)$,  to be the number of non-edges of $W$, i.e. the number of pairs     
$u\neq v \in W$ such that  $uv \not\in E(G)$. 
We  use  $G_W$ to denote the graph obtained from graph $G$ by completing its vertex subset $W$ into a clique. 
 We refer to Diestel's book  \cite{Diestel} for basic definitions from Graph Theory.


\medskip\noindent\textbf{Chordal graphs and minimal triangulations.}
\textit{Chordal} or \textit{triangulated} graphs is the class of graphs containing no induced cycles of length more than three. 
In other words, every cycle of length at least four in a chordal graph contains a chord.
Graph $H = (V,E \cup F)$ is said to be a \textit{triangulation} of $G=(V,E)$ if $H$ is chordal. The triangulation 
$H$ is called \textit{minimal} if $H' = (V, E \cup F')$ is not chordal for every edge subset $F' \subset F$ and $H$ is a \textit{minimum} 
triangulation if $H' = (V, E \cup F')$ is not chordal for every edge set $F'$ such that  $|F'| < |F|$. 
The edge set $F$ for the chordal graph $H$ is called the \textit{fill} of $H$, and if $H$ is a minimum triangulation of $G$, then 
$|F|$ is the minimum fill-in for $G$.

Minimal triangulations can be described in terms of  vertex eliminations (also known as Elimination Game) \cite{Parter61,Fulkerson65}. 
Vertex elimination procedure  takes as input a vertex ordering $\pi \colon  \{1, 2, \dots, n\} \to V(G)$   of graph $G$ and 
outputs a chordal graph $H = H_n$. 
 We put  $H_0 = G$ and define $H_i$ to be the graph obtained from $H_{i-1}$ by completing all neighbours $v$ of $\pi(i)$ in $H_{i-1}$ with  $ \pi^{-1}(v)>i$ into a clique.
 An elimination ordering $\pi$ is called \emph{minimal} if 
the corresponding vertex elimination procedure outputs  a minimal triangulation  of $G$.

\begin{proposition}[\cite{OhtsukiCF76}]\label{char:ordering}
Graph $H$ is a minimal triangulation of $G$ if and only if there exists a minimal elimination ordering $\pi$ of $G$ such that the corresponding procedure outputs  $H$. 
\end{proposition}

We will also need the following  description of the fill edges introduced by vertex eleminations. 
\begin{proposition}[\cite{RoseTL76}]\label{pr:orderFill}
Let $H$ be the chordal graph produced  by vertex elimination of graph $G$ according to ordering  $\pi$. 
Then $uv\not\in E(G)$ is a fill edge of $H$ if and only if there exists a path $P=uw_1w_2 \ldots w_\ell  v$ such that 
$\pi^{-1}(w_i) < \min(\pi^{-1}(u),\pi^{-1}(v))$ for each $1 \leq i \leq \ell$.
\end{proposition}

\medskip\noindent\textbf{Minimal separators.}
Let $u$ and $v$ be two non adjacent vertices of a graph $G$. 
A set of vertices $S \subseteq V$ is an {\em $u,v$-separator} if $u$ and $v$ are in different connected components
of the graph $G[V \sm S]$.  We say that $S$ is a {\em minimal $u,v$-separator} of $G$ if no proper subset of $S$ is an $u,v$-separator and  that $S$ is a {\em minimal separator} of $G$ if there are two vertices $u$ and $v$ such
that $S$ is a minimal $u,v$-separator. Notice that a minimal separator can be
contained  in another one. 
If a minimal separator is a clique, we refer to it as to a \emph{clique minimal separator}. 
A connected component $C$ of $G[V \sm S]$
is a {\em full} component associated to $S$ if $N(C) = S$.
The following proposition is an exercise  in \cite{Golumbic80}.
\begin{proposition}[Folklore]\label{pr:full_components}
A set $S$ of vertices of $G$ is a minimal $a,b$-separator if and
only if $a$ and $b$ are in different full components associated to $S$.
In particular, $S$ is a minimal separator if and only if
there are at least two distinct full components associated to $S$.
\end{proposition}
%


\medskip\noindent\textbf{Potential Maximal Cliques }  are 
 combinatorial objects which properties are crucial for our algorithm. 
A vertex set $\Omega$ is defined as a \textit{potential maximal clique} in graph $G$ 
if there is some minimal triangulation $H$ of $G$ such that $\Omega$ is a maximal clique of $H$.
Potential maximal cliques were defined by Bouchitt\'e and Todinca in   \cite{BouchitteT01,BouchitteT02}.

The following proposition was proved by Kloks et al.  for minimal separators \cite{KloksKS97} and by Bouchitt\'e and Todinca for potential maximal cliques \cite{BouchitteT01}. 

\begin{proposition}[\cite{BouchitteT01,KloksKS97}]\label{char:pmc}
Let $X$ be either a  potential maximal clique or a minimal separator of $G$, and 
let $G_{X} $ be  the graph obtained from $G$  by completing $X$ into a clique. 
Let $C_1,C_2,\ldots, C_r$ be the connected components of $G \sm X$. 
Then graph $H$ obtained from $G_X$ by adding a set of fill edges $F$  is a minimal triangulation of $G$  if and only if 
$F = \bigcup_{i=1}^r F_i$, where $F_i$ is the set of fill edges in a minimal triangulation of $G_{X}[N[C_i]]$. 
\end{proposition}
%



The following result about  the structure of potential maximal cliques is due to
Bouchitt\'e and Todinca.

\begin{proposition}[\cite{BouchitteT01}]\label{pr:pmc_sep}
Let $ \Omega \subseteq V$ be a set of vertices of the graph $G$. 
Let $\{ C_1, C_2, \ldots, C_p\}$ be the set of the connected
components of $G \sm \Omega$ and  $\{ S_1,S_2,\ldots,S_p\}$, where $S_i = N(C_i)$  for  $i \in \{1,2,\ldots ,p\}$. 
Then $\Omega$ is a potential maximal clique of $G$ if and only if:
\begin{enumerate}
\item $G \sm \Omega$ has no full component associated to $\Omega$, and
\item the graph on the vertex set $\Omega$ obtained from $G[\Omega]$ by
completing each $S_i$,  $i \in \{1,2,\ldots,p\}$, into a clique, is a complete graph.
\end{enumerate}
Moreover, if $\Omega$ is a potential maximal clique, then
$\{S_1,S_2,\ldots,S_p\}$ is the set of minimal separators of $G$ contained in $\Omega$.
\end{proposition}

We will need also the following proposition from \cite{FominKTV08}.
\begin{proposition}[\cite{FominKTV08}]\label{pr:vertexRep}
Let $\Omega$ be a potential maximal clique of $G$. Then for every $y \in \Omega$, $\Omega=N_G(Y ) \cup \{y\}$, where 
$Y$ 
is   the connected component of $G \sm (\Omega \sm \{y\})$ containing $y$. 
\end{proposition}

A naive approach of deciding if a given vertex subset is a potential maximal clique would be to try all possible minimal triangulations. There is a much faster approach of recognizing potential maximal cliques  due to Bouchitt\'e and Todinca based on Proposition~\ref{pr:pmc_sep}.

\begin{proposition}[\cite{BouchitteT01}]\label{pr:pmc_rec}
There is an algorithm that, given a graph $G = (V,E)$ and a set of
vertices $ \Omega \subseteq V$, verifies if $ \Omega$ is a potential maximal
clique of $G$ in time 
  $\cO(nm)$.
\end{proposition}
%

%


\noindent\textbf{Parameterized complexity.}
A parameterized problem $\Pi$ is a subset of $\Gamma^{*}\times \mathbb{N}$ for some finite alphabet $\Gamma$. An instance of a parameterized problem consists of $(x,k)$, where $k$ is called the parameter. 
 A central notion in 
parameterized complexity is {\em fixed parameter tractability (FPT)} which means, for a given instance $(x,k)$, 
solvability in time $f(k)\cdot p(|x|)$, where $f$ is an arbitrary function of $k$ and $p$ is a polynomial in the input size. We refer to the book of Downey and Fellows \cite{DowneyF99} for further reading on Parameterized Complexity. 

\noindent\textbf{Kernelization.}
 A \emph{kernelization algorithm} for a  parameterized problem 
$\Pi\subseteq \Gamma^{*}\times \mathbb{N}$ is an algorithm that given $(x,k)\in \Gamma^{*}\times \mathbb{N} $ 
outputs in time polynomial in $|x|+k$ a pair $(x',k')\in \Gamma^{*}\times \mathbb{N}$, called \emph{kernel} such that
$(x,k)\in \Pi$  if and only if $(x',k')\in \Pi$, 
 $|x'|\leq g(k)$,  and 
 $k' \leq k$, where $g$ is some computable function. The function 
$g$ is referred to as the size of the kernel. If~$g(k)=k^{O(1)}$ then we say that $\Pi$ admits a polynomial kernel. 
 
 There are several known polynomial kernels for the  \textsc{Minimum Fill-in} problem  \cite{focs/KaplanST94,KaplanST99}. The best known kernelization algorithm is due to Natanzon et al. \cite{NatanzonSS98,NatanzonSS00}, which for a given instance $(G,k)$ outputs in time $\cO(k^2 nm)$ an instance $(G',k')$ such that $k'\leq k$, $|V(G')|\leq  2k^2 + 4k$, and $(G,k)$ is a YES instance if and only if $(G',k')$ is. 
 \begin{proposition}[\cite{NatanzonSS98,NatanzonSS00}] \label{prop:polykernel}
   \textsc{Minimum Fill-in} has a kernel with vertex set of size $\cO(k^2)$. The running time of the kernelization algorithm is  $\cO(k^2 nm)$.
 \end{proposition} 
 

\section{Branching}\label{sec:branch}

In our algorithm we  apply branching procedure (Rule~\ref{br:rule1}) whenever it is possible.  
To describe this rule, we need some definitions. Let $u,v$ be two nonadjacent vertices of $G$, and let $X=N(u)\cap N(v)$ be the common neighborhood of $u$ and $v$. Let also   $P = uw_1w_2\dots w_\ell v$ be a chordless  $uv$-path. In other words,  any two vertices of $P$ are adjacent if and only if they are consecutive   vertices in $P$.  We say that  \emph{visibility of $X$ from   $P$ is obscured}
 if $ |X \sm N(w_i)|\geq \sqrt{k} $ for every $i\in \{1, \dots, \ell\}$. Thus every internal vertex of $P$ is nonadjacent to at least $\sqrt{k}$ vertices of $X$. 
 
 The idea behind  branching   is based on the observation that  every  fill-in of $G$ with at most $k$ edges should either contain fill  edge $uv$, or should make at least one internal vertex of the path to be adjacent to all vertices of $X$.


\begin{Rule}[Branching Rule]\label{br:rule1}
If   instance $(G=(V,E),k)$ of \textsc{Minimum Fill-in} contains  
a pair of nonadjacent vertices $u,v \in V$ and   a chordless  $uv$-path
 $P = uw_1w_2\dots w_\ell v$  such that   visibility of   $X=N(u) \cap N(v)$ from $P$ is obscure, then branch into $\ell +1$  instances $(G_0,k_0),(G_1,k_1),\ldots, (G_\ell,k_\ell)$. 
 Here 
 \begin{itemize}
 \item $G_0= (V,E\cup \{uv\})$, $k_0 = k-1$;
 \item For $i\in \{1,\dots, \ell\}$, $G_i=  (V,E \cup F_i)$, $k_i = k-|F_i|$, where 
 $F_i=\{w_ix | x\in X   \wedge w_{i}x \not\in E\}$.
 \end{itemize} 
%
\end{Rule}

\begin{lemma}\label{lem:sound_branching}   Rule~\ref{br:rule1} is sound, i.e. $(G,k)$ is a YES instance   if and only if   
$(G_i,k_i)$ is a YES instance for some  $i\in \{0,\dots, \ell\}$. 
\end{lemma}
\begin{proof}
If for some   $i\in \{0,\dots, \ell\}$, $(G_i , k_i)$  is a YES instance, then $G$ can be turned into a chordal graph by adding at most $k_i+ |F_i|=k$ edges, and thus 
$(G,k)$ is a YES instance. 
%


Let $(G,k)$ be a YES instance, and 
let $F\subseteq [V]^2$ be such that graph $H = (V,E \cup F)$ is chordal and $|F| \leq k$.
By Proposition~\ref{char:ordering}, 
 there exists an ordering $\pi$ of $V$, such that the Elimination Game
algorithm on $G$ and $\pi$ outputs $H$. 
Without loss of generality, we can assume that $\pi^{-1}(u) < \pi^{-1}(v)$.
If for some $x \in X$, $\pi^{-1}(x) <  \pi^{-1}(u)$, then by 
Proposition~\ref{pr:orderFill}, we have that  $uv \in F$.  Also  by Proposition~\ref{pr:orderFill}, if  
 $\pi^{-1}(w_i) <  \pi^{-1}(u) $ 
for each  $i\in \{1,\dots, \ell\}$, then again $uv \in F$. In both  cases $(G_0, k_0)$ is a YES instance.

The only remaining case is when  $ \pi^{-1}(u) <\pi^{-1}(x) $ for all $x \in X$, and there  is at least one vertex of $P$ placed after $u$  in  ordering $\pi$.   
Let   $i\geq 1$ be the smallest index such 
that $\pi^{-1}(u) < \pi^{-1}(w_i)$. Thus for every $x\in X$,  in the path $xuw_1,w_2, \ldots, w_{i}$ all
internal vertices   are ordered by $\pi$ before $x$ and $w_i$. By Proposition~\ref{pr:orderFill}, this imply that $w_i$ is adjacent to all vertices of $X$, and hence $(G_i, k_i)$ is a YES instance.   
%
\end{proof}

The following lemma shows that every branching step of Rule~\ref{br:rule1} can be performed in polynomial time.

\begin{lemma}\label{lem:runningTime} Let $(G,k)$ be an instance of $\textsc{Minimum Fill-in}$.
It can be identified in time $\cO(n^4)$ if  there is a pair $u,v\in V(G)$ satisfying  conditions of Rule~\ref{br:rule1}.  Moreover, if conditions of Rule~\ref{br:rule1} hold, then  a  
 pair  $u,v$ of two nonadjacent vertices   and a chordless   $uv$-path $P$ such that visibility of $N(u)\cap N(v)$ from $P$ is obscured,  can   be found in time $\cO(n^4)$. 
%
%
%

\end{lemma}
\begin{proof}
For each  pair of nonadjacent vertices $u,v$, we compute  $X = N(u) \cap N(v)$. We compute the set of all vertices  $W\subseteq V(G)\sm  \{u,v\}$,   such that every vertex of $W$ is nonadjacent to at least $\sqrt{k}$ vertices of $X$. 
Then conditions of Rule~\ref{br:rule1} do not hold for $u$ and $v$ if 
  in the subgraph $G_{uv}$ induced by $W\cup  \{u,v\}$,   $u$ and $v$ are in different connected components.  
If $u$ and $v$ are in the same connected component of $G_{uv}$, then a  shortest (in  $G_{uv}$) $uv$-path $P$ is a chordless path  and the visibility of $X$ from $P$ is obscured. Clearly, all these procedures can be performed in  time $\cO(n^4)$.
\end{proof}

We say that instance $(G,k)$ is   \emph{non-reducible} if conditions of Rule~\ref{br:rule1} do not hold. Thus for each pair of vertices $u,v$ of non-reducible graph $G$, there is no $uv$-path with   obscure visibility of $N(u)\cap N(v)$.

\begin{lemma} \label{lem:amountof_reduced}
Let $t(n,k)$ be the maximum number of non-reducible problem instances 
  resulting  from recursive application of   Rule \ref{br:rule1}  starting from instance $(G,k)$ with $|V(G)|=n$.
Then 
 $t(n,k)=n^{\cO(\sqrt{k})}$
 and all generated non-reducible instances can be enumerated within the 
same time bound. 
\end{lemma}
\begin{proof}
Let us assume that we branch on the instances corresponding to a   pair $u,v$ and path $P = uw_1w_2\dots w_\ell v$    such that the visibility of $N(u)\cap N(w)$ is obscure from $P$. 
Then
the value of $t(n,k)$ is at most 
$
\sum_{i=0}^\ell t(n,k_i). 
$
 Here $k_0=1$ and for all $i\geq 1$, $k_i= k-|F_i|\leq k-\sqrt{k}$. Since the number of vertices in $P$ does not exceed $n$, we have that 
$
 t(n,k)\leq t(n,k-1) + n \cdot t(n,k-\sqrt{k}).
$
 By making use of  standard arguments on the amount of leaves in  branching trees (see, for example \cite[Theorem 8.1]{JurdzinskiPZ08})
  it follows  that  $t(n,k)=n^{\cO(\sqrt{k})}$.    By Lemma~\ref{lem:runningTime}, every recursive call of the branching algorithm can be done in time  $\cO(n^4)$, and thus  all  non-reducible instances are  generated in time $\cO(n^{\cO(\sqrt{k})} \cdot n^4)=n^{\cO(\sqrt{k})}$.

\end{proof}

\section{Listing vital potential maximal cliques} \label{sec:pmc}
Let $(G,k)$ be a YES instance of  \textsc{Minimum Fill-in}. It means that $G$ can be turned into a chordal graph $H$ by adding at most $k$ edges.  Every maximal clique in $H$ corresponds to a potential maximal clique of $G$. The observation here is that if a potential maximal clique $\Omega$ needs more than 
$k$ edges to be added to become  a clique, then no solution $H$ can contain $\Omega$ as a maximal clique.  
In Section~\ref{ss:dp} we prove that 
  the only potential maximal cliques that are essential for   a fill-in with at most $k$ edges are the ones that miss at most $k$ edges from a clique.


A potential maximal clique $\Omega$ is \emph{vital}
if the number of edges in $G[\Omega]$ is at least $|\Omega| (|\Omega| -1)/2 -k$.
 In other words, the subgraph induced by vital potential maximal clique   can be turned into 
a complete graph by adding at most $k$ edges. 
In this section we show that all vital potential maximal cliques  of an $n$-vertex 
non-reducible graph  can be enumerated in time $n^{\cO(\sqrt{k})}$. 

We will first show how to enumerate potential maximal cliques  which are, in some sense, almost cliques. This enumeration algorithm will be used as a subroutine to enumerate vital potential maximal cliques. 
A potential maximal clique $\Omega$ is   \emph{quasi-clique}, 
if there is a set $Z \subseteq \Omega$ of size at most $5{\sqrt{k}}$ such that $\Omega \sm Z$ is a clique. 
In particular, if $|\Omega|\leq 5{\sqrt{k}}$, then $\Omega$ is also a quasi-clique. 
The following   lemma gives an algorithm enumerating all quasi-cliqes. 

\begin{lemma}\label{lem:almostcliques}
Let $(G,k)$ be a   problem instance on $n$ vertices. 
Then all quasi-cliques in $G$  can be enumerated within time  $n^{\cO(\sqrt{k})}$.
\end{lemma}

\begin{proof}
We will prove  that 
while  a quasi-clique can be very large, it  can be 
 reconstructed in polynomial time from a small set of $ \cO(\sqrt{k})$ vertices.  
Hence  all quasi-cliques can be generated by enumerating   vertex subsets of size $ \cO(\sqrt{k})$.  
Because the amount of vertex subsets of size $ \cO(\sqrt{k})$  is $n^{\cO(\sqrt{k})}$, this will prove the lemma.

Let $\Omega$ be a \pmc \, which is  a {quasi-clique}, and 
let $Z \subseteq \Omega$ be the set of size at most $5{\sqrt{k}}$ such that   $X = \Omega \sm Z$ is a clique. 
Depending on the amount of full components associated to $X$ in $G \sm (Z \cup X)$, we consider three cases: There are at least two full components, there is exactly one, and there is no full component. 
%

 Consider first the case when $X$ has at least two full components, say  $C_1$ and $C_2$. 
  In this case, by Proposition~\ref{pr:full_components},  $X$ is a minimal clique separator of $G \sm Z$. 
Let $H$ be some  \emph{minimal} triangulation of $G \sm Z$.  
By Proposition~\ref{char:pmc}, clique minimal separators remain clique minimal separators  in every minimal triangulation. Therefore,   $X$ is a minimal separator in $H$.  It is well known that every chordal graph has 
  at most $n-1$ minimal separators  and that they  can be enumerated in linear time \cite{ChandranG06}. 
  To enumerate quasi-cliques   we implement the following algorithm. 
  We construct a minimal triangulation $H$ of  $G \sm Z$.  A  minimal triangulation    can be  constructed in time $\cO(nm)$ or 
$\cO(n^{\omega} \log{n})$, where  $\omega < 2.37$ is the exponent of matrix multiplication and   $m$ is the number of edges in $G$ \cite{HeggernesTV05,RoseTL76}. 
For every minimal separator $S$ of $H$, where $G[S]$ is a clique, we check if $S \cup Z$ is a potential maximal 
clique in $G$. This   can be done in $\cO(km)$ time by Proposition~\ref{pr:pmc_rec}. 
 Therefore, in this case, the  time required to enumerate all quasi-cliques of the form $X\cup Z$, up to polynomial multiplicative factor is proportional to the amount of  sets $Z$ of size at most $5\sqrt{k}$.
 The total running time to enumerate quasi-cliques  of this type is  $n^{\cO(\sqrt{k})}$.

\medskip
Now we consider the case  when no full component in $G\sm (Z \cup X)$ associated to $X$. 
It means that  for every connected component $C$ of $G \sm (Z \cup X)$, there is  $x \in X \sm N(C)$. 
By Proposition~\ref{pr:pmc_sep}, $X$ is also a potential maximal clique in $G\sm Z$.
We construct a  minimal triangulation $H$ of $G \sm Z$. By Proposition~\ref{char:pmc}, 
$X$ is also a potential maximal clique in $H$.
By the  classical result of Dirac \cite{Dirac61} chordal graph $H$ contains at most $n$ maximal cliques 
and   all the maximal cliques of $H$  can be enumerated in linear time \cite{BlairP93}. 
For every maximal clique $K$ of $H$ such that $K$ is also a  clique in $G$, we check if $K \cup Z$ is a potential maximal 
clique in $G$, which can be done in $\cO(nm)$ time by Proposition~\ref{pr:pmc_rec}. 
  As in the previous case, the enumeration of all such quasi-cliques   boils down to enumerating sets $Z$, which takes time  $n^{\cO(\sqrt{k})}$.

%

%
\medskip 

Last case, vertex set $X$ has  unique full component $C_r$ in $G \sm (Z \cup X)$ associated to $X$.
Since $\Omega = Z \cup X$, we have that each of the connected components $C_1,C_2,\ldots,C_r$ 
of $G\sm (Z\cup X)$ is also  a connected component of $G \sm \Omega$.
Then  for every $i\in \{1,\dots, r-1\}$, $S_i = N_{G \sm Z}(C_i)$ is a clique minimal separator in $G \sm Z$ because  $S_i \subseteq X$ is a clique, and  $C_i$  together with the component of $G \sm (Z \cup S_i)$
containing $X \sm S_i$,  are full components associated to $S_i$. 
Let $H$ be a minimal triangulation of $G \sm Z$. 
Vertex set $X$ is a clique  in $G \sm Z$ and thus is  a clique  in $H$. 
Let $K$ be a maximal clique of $H$ containing   $X$. 
By Proposition~\ref{char:pmc},    for every $i\in \{1,\dots, r-1\}$, $S_i$ is a minimal separator in $H$. 
By Proposition~\ref{char:pmc},   $G \sm Z$ has no
 fill edges between vertices separated by $S_i$  and thus 
$C_i$ is a connected component of $H \sm K$. 


Because $\Omega$ is a \pmc  \, in $G$,   by Proposition~\ref{pr:pmc_sep}, there is   $y\in \Omega$
such that $y \not\in N_G(C_r)$. Since $C_r$ is a full component for $X$, we have that $y\in Z$. 
Moreover, every connected component $C\neq C_r$ of $G \sm (Z\cup X)$ is also a connected component of 
$H \sm K$. Thus every 
connected component of $H \sm K$  containing a neighbor of  $y$  in $G$ is also a connected component    of $G \sm \Omega$ containing a neighbor of $y$. 
%

Let $B_1, B_2, \dots, B_\ell$ be the set of connected components
in $G \sm (K \cup Z)$ with $y\in N_G(B_i)$. We define 
\[
Y=\bigcup_{1\leq i \leq \ell} B_i  \cup\{y\}.
\]


By Proposition~\ref{pr:vertexRep},  
$\Omega = N_G(Y ) \cup \{y\}$ and in this case potential maximal clique is characterized by $y$ and $Y$.

To summarize, to enumerate all quasi-cliques corresponding to the last case, we do the following. For every set $Z$ of size at most $5\sqrt{k}$, we construct a minimal triangulation $H$ of $G\setminus Z$. The amount of maximal cliques in a chordal graph $H$ is at most $n$, and for every maximal clique $K$ of $H$ and for every 
$y\in Z$, we compute the set $Y$.  We use Propositions~\ref{pr:pmc_rec} to check if $N_G(Y ) \cup \{y\}$ is a potential maximal clique. 
The total running time to enumerate quasi-cliques in this case is bounded, up to polynomial factor, by  the amount of subsets of size $\cO(\sqrt{k})$ in $G$ which is  $n^{\cO(\sqrt{k})}$.
\end{proof}

Now we are ready to prove the result about  vital  potential maximal cliques in non-reducible graphs.

\begin{lemma}\label{lem:essent_non_red}
Let $(G,k)$ be a non-reducible instance of the problem. 
All vital  potential maximal cliques in $G$ can be enumerated within time  $n^{\cO(\sqrt{k})}$, where $n$ is the number of vertices in $G$. 
\end{lemma}
\begin{proof}
We start by enumerating all vertex subsets of $G$ of size at most $5\sqrt{k} +2$ and apply Proposition~\ref{pr:pmc_rec} 
to check if each such set is a vital potential maximal clique or not. 

Let $\Omega$ be a vital  potential maximal clique with at least $5\sqrt{k} +3$ vertices  
and let $Y \subseteq \Omega$ be the set of vertices of $\Omega$ such that each vertex of $ Y$ is adjacent in $G$ 
to at most $|\Omega|-1 -\sqrt{k}$ vertices of $\Omega$. To turn $\Omega$ into a complete graph, 
for each vertex  $v\in Y$,  we have to add at least $\sqrt{k}$ fill  edges incident to $v$. Hence $|Y| \leq 2\sqrt{k}$.
If  $\Omega \sm Y$ is a clique, then $\Omega$ is a quasi-clique. 
By Lemma~\ref{lem:almostcliques}, all quasi-cliques can be enumerated in time $n^{\cO(\sqrt{k})}$.

If  $\Omega \sm Y$ is not a clique, there is at least one pair of non-adjacent vertices  $u,v \in \Omega \sm Y$.
By Proposition~\ref{pr:pmc_sep}, there is a connected component $C$ of $G \sm \Omega$ such that $u,v \in N(C)$. 

\begin{claim}\label{claim:comp}
There is   $w\in C$ such that $ |\Omega\sm N(w)|\leq 5 \sqrt{k}+2$.
\end{claim}
\begin{proof}
Targeting towards a contradiction, we 
   assume that  the claim  does not hold. 
 We define the following subsets of $\Omega \sm Y$.  

\begin{figure}
\begin{center}
 \includegraphics[width=3.5cm]{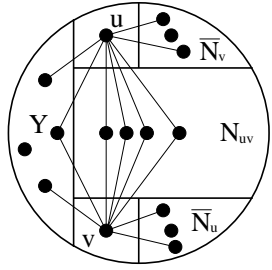}
\end{center}
\caption{Partitioning of potential maximal clique $\Omega$   into   sets
$\overline{N}_u , \overline{N}_v, N_{uv} , \{u\} , \{v\} ,$ and $Y$. }  
\label{fig:pmc_part}
\end{figure} 
 \begin{itemize}
\item  $\overline{N}_u  \subseteq \Omega\sm Y$ is the set of vertices which are not adjacent to $u$, 
\item  $\overline{N}_v \subseteq \Omega\sm  Y$ is the set of vertices which are not adjacent to $v$, and 
\item $N_{uv} = \Omega\sm( Y\cup \overline{N}_u \cup\overline{N}_v )$ is the set of vertices adjacent to $u$ and to $v$.  
\end{itemize}
 See   Fig.~\ref{fig:pmc_part} for an illustration.
Let us note that  
\[
\Omega= \overline{N}_u \cup \overline{N}_v \cup N_{uv} \cup \{u\} \cup \{v\} \cup  Y.
\]
Since $u,v \not\in Y$, we have that  $\max\{|\overline{N}_u|, |\overline{N}_v|\}\leq \sqrt{k}$. 

We claim that $|N_{uv}| \leq  \sqrt{k}$. Targeting towards a contradiction, let us assume that  $|N_{uv}| >\sqrt{k}$.
By our assumption, every vertex $w \in C$ is not adjacent to at least $5\sqrt{k} +2$ vertices of $\Omega$.
Since $|Y \cup \overline{N}_u \cup \overline{N}_v \cup \{u\}  \cup \{v\}| \leq 2\sqrt{k}+\sqrt{k}+\sqrt{k}+2=4\sqrt{k}+2$, we have that 
each vertex of $C$ is nonadjacent to at least $\sqrt{k}$ vertices of $N_{uv}$. 
We take a shortest $uv$-path $P$ with all internal vertices in $C$. 
Because $C$ is a connected component and $u,v\in N(C)$, such a path exists. 
Every internal vertex of $P$ is nonadjacent to at least $\sqrt{k}$ vertices of $N_{uv}\subseteq N(u)\cap N(v)$, and thus the visibility of  $N_{uv}$ from $P$ is obscured. 
But this is a contradiction to the assumption that $(G,k)$ is non-reducible.  Hence $|N_{uv}| \leq \sqrt{k}$.
 
Thus if the claim does not hold, we have that  
\[
|\Omega|= |\overline{N}_u \cup \overline{N}_v \cup N_{uv} \cup \{u\} \cup \{v\} \cup  Y| \leq 5\sqrt{k}+2,
\]
but this contradicts to the assumption that $|\Omega|\geq \ 5\sqrt{k}+3$. This concludes the proof of the claim.  
\end{proof}

%
%
  
\medskip
We have shown that for every 
 vital  potential maximal clique $\Omega$  of size  at least $5\sqrt{k} +3$, there is a connected component $C$  and 
 $w \in C$  such that $|\Omega \sm N(w)|\leq 5\sqrt{k}+2$.
Let $H$ be the graph obtained from $G$ by completing $N(w)$ into a clique. 
The graph $H[\Omega]$ consist of a clique plus at most $5\sqrt{k}+2$ vertices.
We want  to show that  $\Omega$ is a quasi-clique in $H$, by arguing  that $\Omega$ is a potential maximal clique in $H$.
 Vertex set $\Omega$ is a potential maximal clique in $G$, and thus by Proposition~\ref{pr:pmc_sep}, 
there is no full component associated to $\Omega$ in $G \sm \Omega$. 
Because 
$N(w)\cap  \Omega \subseteq N(C)\subset \Omega$, 
there is no full  component associated to $\Omega$ in $H$. 
  We   use Proposition~\ref{pr:pmc_sep} to show that $\Omega$ is a potential maximal clique in $H$ as well. 
Hence $\Omega$ is a quasi-clique in $H$.

%
\medskip

To conclude, we use the following strategy to enumerate all vital potential maximal cliques. 
We enumerate first all quasi-cliques in $G$ in time  $n^{\cO(\sqrt{k})}$ by making use of Lemma~\ref{lem:almostcliques},  and for each quasi-clique we  use Proposition~\ref{pr:pmc_rec} to check if it is a vital potential maximal clique. We also try all vertex subsets of size at most $5\sqrt{k}+2$ and check if each of such sets is a vital potential maximal clique.  All vital potential maximal cliques which are not enumerated prior to this moment should satisfy the condition   of the claim.  As we have shown, each such vital potential maximal clique is a quasi-clique in the graph $H$ obtained from $G$ by selecting some vertex $w$ and turning  $N_G(w)$ into clique. Thus 
 for every vertex $w$ of $G$, we construct graph $H$  
and then use Lemma~\ref{lem:almostcliques} to enumerate all quasi-cliques in $H$. For each quasi-clique of $H$, 
we use Proposition~\ref{pr:pmc_rec} to check if it is a vital potential maximal clique in $G$. The total running time of this procedure is   $n^{\cO(\sqrt{k})}$.
\end{proof}

\section{Exploring   remaining solution space}\label{ss:dp}

For an instance $(G,k)$ of  \textsc{Minimum Fill-in}, let $\Pi_k$ be the set of all vital potential maximal cliques. 
 In this section, we give an algorithm of running time  $\cO(nm |\Pi_k|)$, where $n$ is the number of vertices and $m$ the number of edges in $G$. The algorithm receives $(G,k)$ and $\Pi_k$ as an input, and decides   if 
 $(G,k)$ is a YES instance. 
 The  algorithm is a modification of the algorithm from  \cite{FominKTV08}. The only difference is that the algorithm from \cite{FominKTV08} computes   an  optimum triangulation from  the set of all potential maximal cliques while here we have to work only with vital potential maximal cliques. 
For reader's convenience we provide the full proof, but first we need the following lemma.

\begin{lemma}\label{pr:pmc_SC}
Let $S$ be a minimal separator in $G$ and 
let $C$ be a full connected component of $G \sm S$ associated to $S$. 
Then every minimal triangulation $H$ of $G_S$ contains a maximal clique $K$ 
such that $S \subset K \subseteq S \cup C$.
\end{lemma}
\begin{proof}
By Proposition~\ref{char:pmc},  $H$ is a minimal triangulation of $G_S$ if and only if 
$H[S \cup C]$ is a minimal triangulation of $G_S[S \cup C]$. 
Because  $S$ is a clique $G_S[S]$ in, $S$ is a subset of some maximal clique $K$ of 
 $H[S \cup C]$. 
By definition $K$ is a potential maximal clique in $G_S[S \cup C]$, and by Proposition~\ref{pr:pmc_sep} 
$K$ is a potential maximal clique in $G$. 
Since $G_S[S \cup C] \sm S$ has a full component associated to $S$, we have that  by Proposition~\ref{pr:pmc_sep},
$S$ is not a potential maximal clique in $G_S[S \cup C]$ and thus $S \subset K$. 
\end{proof}


\begin{lemma}\label{lem:pmc_fill} Given a set of all  vital potential maximal cliques 
 $\Pi_k$ of $G$,  it can be decided in time $\cO(nm|\Pi_k|)$ if   $(G,k)$  is a YES instance of \textsc{Minimum Fill-in}. 
\end{lemma}

\begin{proof}

Let $\textbf{mfi}(G)$ be the minimum number of fill edges needed to triangulate $G$.
Let us remind that by $ \textbf{fill}_G(\Omega)$ we denote the number of non-edges in $G[\Omega]$   and by $G_\Omega$ the graph obtained from $G$ by completing $\Omega$ into a clique. 
If $\textbf{mfi}(G)\leq k$,  
then by  Proposition~\ref{char:pmc},  we have that 
\begin{equation}\label{eq:root}
 \textbf{mfi}(G) = \min_{\Omega \in \Pi_k}[  \textbf{fill}_G(\Omega) + \sum_{C \text{ is a component of }G \sm \Omega} \textbf{mfi}(G_{\Omega}[C \cup N_G(C)])].
\end{equation}
Formula \eqref{eq:root} can be used to compute $ \textbf{mfi}(G)$, however by making use of this formula we are not able to obtained the claimed running time. 
To implement the algorithm in  time $\cO(nm |\Pi_k|)$, we  compute  
$\textbf{mfi}(G_{\Omega}[C \cup N_G(C)])]$ by dynamic programming.

By Proposition~\ref{pr:pmc_sep},  for every  connected component $C$ of $G \sm \Omega$, where $\Omega\in \Pi_k$,  $S = N_G(C)\subset \Omega$ is a minimal separator. We define the set $\Delta_k$ as the set of all minimal separators $S$, such that $S=N(C)$ for some 
connected component $C$ in $G \sm \Omega$ for some  $\Omega\in \Pi_k$. Since for every   $\Omega\in \Pi_k$ the number of components in $G \sm \Omega$ is at most $n$, we have that $|\Delta_k| \leq n |\Pi_k|$.

For $S\in \Delta_k$ and a full connected component $C$ of $G\sm S$ associated to $S$. 
We define $\Pi_{S,C}$ as the set of potential maximal cliques  $\Omega \in \Pi_k$ such that 
   $S \subset \Omega \subseteq S \cup C$.
  The triple  $(S, C ,\Omega)$ was called a good triple in \cite{FominKTV08}.

For 
every $\Omega\in \Pi_k$, connected component $C$ of $G\sm \Omega$, and $S=N(C)$, we 
 compute $\textbf{mfi}(F)$, where $F= G_{\Omega}[C \cup S]$.
 We start dynamic programming by  computing the values for all  sets
 $(S,C)$ such that $\Omega=C\cup S$ is an inclusion-minimal potential maximal clique. In this case we put 
  $\textbf{mfi}(F)=\textbf{fill}(C\cup S)$.
Observe that $G_{S}[C \cup S] =  G_{\Omega}[C \cup S]$. Hence  by Lemma~\ref{pr:pmc_SC},
for every minimal triangulation  $H$   of $G_S$, there exists a potential maximal clique $\Omega$ in $G$ 
such that $\Omega$ is a maximal clique in $H$ and $S \subset \Omega \subseteq S \cup C$.
Thus $\Omega \in \Pi_{S,C}$. Using this observation, we write the following formula for dynamic programming. 
  
\begin{equation}\label{eq:recursion}
\textbf{mfi}(F) = \min_{\Omega' \in \Pi_{S,C}}[ \textbf{fill}_{F}(\Omega') + \sum_{C' \text{ is a component of } F \sm \Omega'} \textbf{mfi}(F_{\Omega'}[C' \cup N(C')])].
\end{equation}

The fact   $S \subset \Omega'$  ensures that the solution 
in \eqref{eq:recursion} can be reconstructed from instances  with 
$|S \cup C|$ of smaller sizes.  
By  \eqref{eq:root} and \eqref{eq:recursion} we can decide if there exists a triangulation 
of $G$ using at most $k$ fill edges, and to construct such a triangulation. It remains to argue for the running time. 
%

\medskip

Finding connected components in $G \sm \Omega$ and computing $\textbf{fill}(\Omega)$ 
can easily been done in $\cO(n+m)$ time. 
Furthermore,  \eqref{eq:root} is applied $|\Pi_k|$ times  in total.
The running time of dynamic programming using \eqref{eq:recursion} is proportional to the amount of states of dynamic programming, which is 
\[
\sum_{S \in \Delta_k} \sum_{C \in G \sm S} |\Pi_{S,C}| 
\]
 The graph $G \sm \Omega$ contains at most $n$ connected components and thus for every minimal separator, each potential maximal clique is counted at   most $n$ times, and thus the amount of the elements in the sum does not exceed  $ n|\Pi_k|$.
%
The total running time is $\cO(nm|\Pi_k|)$.
\end{proof}

\section{Putting things together}\label{sec:puttingtogether}
Now we are in the position to prove the main result of this paper. 
\begin{theorem}\label{thm:mainthm} The \textsc{Minimum Fill-in} problem is solvable in time 
$\cO(2^{\cO(\sqrt{k}\log{k})} +k^2 nm)$.
\end{theorem}
\begin{proof}
\smallskip\noindent\textbf{Step A.}
Given instance $(G,k)$ of the  \textsc{Minimum Fill-in} problem, we use Proposition~\ref{prop:polykernel}   to obtain a kernel $(G', k')$ on $\cO(k^2)$ vertices and with $k'\leq k$. Let us note that $(G,k)$ is a YES instance if and only if 
 $(G',k')$ is a YES instance.
This step is performed in time $\cO(k^2 nm)$. 

\smallskip\noindent\textbf{Step B1.}
We use Branching Rule~\ref{br:rule1} on instance $(G', k')$. Since the number of vertices in $G'$ is $\cO(k^2)$, we have that 
by  Lemma~\ref{lem:amountof_reduced}, the result of this procedure is  the set of $(k^2)^{\cO(\sqrt{k})}=2^{\cO(\sqrt{k}\log{k})}$ non-reducible instances $(G_1, k_1), \dots, (G_p, k_p)$. For each $i\in \{1, 2, \dots, p\}$, graph $G_i$ has
$\cO(k^2)$ vertices and $k_i \leq k$. Moreover,  by Lemma~\ref{lem:sound_branching},   $(G',k')$, and thus $(G,k)$, is a YES instance if and only if at least one $(G_i, k_i)$ is a YES instance.  By Lemma~\ref{lem:amountof_reduced}, the running time of this step is $2^{\cO(\sqrt{k}\log{k})}$.

\smallskip\noindent\textbf{Step B2.} For each $i\in \{1, 2, \dots, p\}$, we list all vital potential maximal cliques of graph $G_i$. 
By Lemma~\ref{lem:essent_non_red}, the amount of all vital potential maximal cliques in non-reducible graph $G_i$ is $2^{\cO(\sqrt{k}\log{k})}$ and they can be listed within the same running time. 

\smallskip\noindent\textbf{Step C.} At this step   for each $i\in \{1, 2, \dots, p\}$, we are given instance $(G_i, k_i)$ together with the set $\Pi_{k_i}$ of vital potential maximal cliques of $G_i$ computed in   Step~B2. We use Lemma~\ref{lem:pmc_fill} to solve  the  \textsc{Minimum Fill-in} problem for instance $(G_i, k_i)$ in time $\cO(k^6|\Pi_{k_i}|)=2^{\cO(\sqrt{k}\log{k})}$. If at least one of the instances  $(G_i, k_i)$ is a YES instance, then by Lemma~\ref{lem:sound_branching},   $(G, k)$ is a YES instance. If all instances $(G_i, k_i)$ are  NO instances,   we conclude that $(G,k)$ is a NO instance. 
Since $p=2^{\cO(\sqrt{k}\log{k})}$, we have that Step~C can be performed in time 
$2^{\cO(\sqrt{k}\log{k})}$.
The total running time required to perform all steps of the algorithm is $\cO(2^{\cO(\sqrt{k}\log{k})} +k^2 nm)$.

\end{proof}

Let us remark that  our decision  algorithm can be easily adapted to output the optimum fill-in of size at most $k$.


\section{Applications  to other problems}\label{sec:other_problems}

The algorithm described in the previous section can be modified  to solve several  related  problems. Problems considered in this section are 
\textsc{Minimum Chain Completion}, \textsc{Chordal Graph Sandwich},
and \textsc{Triangulating Colored Graph}.

\medskip\noindent\textbf{Minimum Chain Completion.}
Bipartite graph $G=(V_1,V_2,E)$ is a chain graph if the neighbourhoods of the nodes in $V_1$ forms a chain, that is  
there is an ordering $v_1,v_2,\ldots, v_{|V_1|}$ of the vertices in $V_1$, such that 
$N(v_1)\subseteq N(v_2) \subseteq \ldots \subseteq N(v_{|V_1|})$.

\begin{center}
\fbox{\begin{minipage}{13cm}
\noindent {\sc Minimum Chain Completion }\\
{\sl Input:} A bipartite graph $G = (V_1,V_2,E)$ and  integer $k\geq 0$.\\
{\sl Question:} Is there $F \subseteq  V_1 \times V_2$, $|F|\leq k$, such that graph $H=(V_1,V_2,E \cup F)$ 
a chain graph? 

\end{minipage}}
\end{center}

 Yannakakis in his famous  NP-completeness proof of \textsc{Minimum Fill-in} ~\cite{Yannakakis81} used the following observation. 
Let $G$ be a bipartite graph with bipartitions $V_1$ and $V_2$, and let $G'$ be cobipartite (the complement of bipartite) graph formed by turning $V_1$ and $V_2$ into cliques. Then $G$ can be transformed into a chain graph by adding $k$ edges if and only if $G'$ can be triangulated with $k$ edges. By Theorem~\ref{thm:mainthm}, we have that
\textsc{Minimum Chain Completion}  is solvable in $\cO(2^{\cO(\sqrt{k}\log{k})} +k^2 nm)$ time.

%

\medskip\noindent\textbf{Chordal graph sandwich.}
In the chordal graph sandwich problem we are given two graphs $G_1 = (V,E_1)$ and $G_2 = (V,E_2)$ 
on the same vertex set $V$, and with  $E_1 \subset E_2$. 
The \textsc{Chordal Graph Sandwich} problem asks if there exists a chordal graph $H=(V,E_1 \cup F)$ 
sandwiched in between $G_1$ and $G_2$, that is $E_1 \cup F \subseteq E_2$.
\begin{center}
\fbox{\begin{minipage}{13cm}
\noindent {\sc Chordal Graph Sandwich}\\
{\sl Input:} Two graphs $G_1 = (V,E_1)$ and $G_2=(V,E_2)$ such that $E_1 \subset E_2$, and   $k = |E_2 \sm E_1|$.\\
{\sl Question:} Is there $F \subseteq E_2 \sm E_1$, $|F|\leq k$, such that graph $H=(V,E_1 \cup F)$ is a  triangulation of $G_1$? 
\end{minipage}}
\end{center}

Let us remark that  the \textsc{Chordal Graph Sandwich} problem is equivalent to asking if there is   a minimal triangulation of $G_1$ 
sandwiched between $G_1$ and $G_2$.

To solve
{\sc Chordal Graph Sandwich}
 we cannot use the algorithm from Section~\ref{thm:mainthm} directly.  
The reason  is that we are only allowed to add edges from  $E_2$ as fill edges. We need a kernelization algorithm for this problem as well. 
 
 \begin{lemma}\label{lem:kernel_sandwich}
 {\sc Chordal Graph Sandwich}  has a kernel with vertex set of size   $ \cO(k^3)$. The running time of kernelization algorithm is  $\cO(k^2 nm)$.
  \end{lemma}

\begin{proof}
 To simplify notations, we denote an instance of the  \textsc{Chordal Graph Sandwich}  problem by
$(G,E',k)$, where 
$G_1 = G=(V,E)$, 
$E' \cap E = \emptyset$, 
$k \leq |E'|$,
and $G_2 = (V,E \cup E')$. 

We first define two reduction rules and prove their correctness. 

\medskip
\noindent\textbf{No-cycle-vertex Rule.}
\textsl{If  instance $(G,E',k)$  has  $u\in V$ such that   for each connected component $C$ of $G[V \sm N[u]]$,
$N(C)$ is a clique, then
replace instance $(G,E',k)$ with instance $(G \sm \{u\},E'',k')$, where 
  $E''\subseteq E'$ are the edges not incident to $u$ and  $k' = k- |E' \sm E''|$.
}

\begin{claim}\label{lem:sound_rr:no_cycle} 
No-cycle-vertex Rule  is sound, i.e. $(G \sm \{u\},E'',k')$ is 
an YES instance of the chordal sandwich problem if and only if $(G,E',k)$ is an YES instance.
\end{claim}
\begin{proof}
Let $G_u$ be the graph $G[V \sm \{u\}]$. 
Chordallity is a hereditary property so if $H = (V,E \cup F)$ is a triangulation of $G$ where $|F| \leq k$, 
then $H[V \sm \{u\}]$ is a triangulation of $G_u$ where the set of fill edges are the edges of $F$ not incident to $u$.

For the opposite direction assume that $H_u = (V(G_u),E(G_u) \cup F_u)$ is a minimal triangulation of $G_u$ where $|F_u| \leq k$.
Our objective now is to argue that $H = (V,E \cup F_u)$ is a triangulation of $G$.
Targeting towards a contradiction,  let us assume that $W=uaw_{1}w_{2}\ldots w_{\ell}bu$ is a chordless cycle of length at least four in $H$. Then  
notice that $u$ is a vertex of $W$ as $H_u = H[V \sm \{u\}]$ is a chordal graph by our assumption. 
Vertices $a$ and $b$ of $W$ are not adjacent by definition, and 
let vertices $w_{1}w_{2}\ldots w_{\ell}$ be contained in the connected component $C$ of $G[V \sm N[u]$. 
Nonadjacent vertices $a,b$ are now contained in $N(C)$ which is a contradiction to the condition of applying No-cycle-vertex Rule.
\end{proof}

For each pair of nonadjacent vertices  $x,y \in V$,  we define $A_{x,y}$ as the set of vertices such that $w \in A_{x,y}$ 
if $x,y \in N_G(w)$ and vertices $x,y$ are contained in the same connected component of $G[(V \sm N[w]) \cup \{x,y\}]$.

\medskip
\noindent\textbf{Safe-edge Rule.}
\textsl{If $2k < |A_{x,y}|$ for some pair of vertices $x,y$ in a problem instance $(G,E',k)$ then
\begin{itemize}
\item replace instance $(G,E',k)$ with a trivial NO instance if $xy \not\in E'$, and 
\item 
otherwise with instance $(G=(V,E \cup \{xy\}),E' \sm \{xy\},k-1)$.
\end{itemize}
}
\begin{claim}\label{lem:sound_rr:safe_edge} 
Safe-edge Rule is sound, i.e. the instance outputted by the rule  is a YES 
 instance of the chordal sandwich problem if and only if $(G,E',k)$ is an YES instance.
\end{claim}
\begin{proof}
By the definition of $A_{x,y}$, it is not hard to see that there exists an induced cycle of length at least four 
consisting of $x,w,y$ and a shortest induced path from $x$ to $y$ in $G[(V \sm N[w]) \cup \{x,y\}]$.
A trivial observation is that every triangulation of $G$ either has 
$xy$ as a fill edge, or there exists a fill edge incident to $w$. 
Since $2k < |A_{xy}|$, we have that every minimal triangulation not using $xy$ as a fill edge has at least 
one fill edge incident to each vertices in $A_{uv}$, and thus $(G,E',k)$ is a NO instance if $xy \not\in E'$ and 
$xy \in F$ for every edge set $F \subseteq E'$ such that $H=(V,E \cup F)$ is chordal. 
\end{proof}

It is possible to show that exhausting application of both No-cycle-vertex and Safe-edge rules 
either results in a polynomial kernel or a trivial recognizable NO instance. 
However, the running time of the algorithm would be $\cO(kn^2 m)$. 
In what follow, we show that with much more careful implementation of rules, 
it is possible to obtain  a kernel of size $\cO(k^3)$ in time $\cO(k^{2}nm)$. 
The algorithm uses the same approach as the kernel algorithms given 
in \cite{KaplanST99} and \cite{NatanzonSS00} for \textsc{Minimum Fill-in}.

Let $(G,E',k)$ be an instance of the problem. We give
an algorithm with running time $\cO(k^{2}nm)$ that outputs 
an instance $(G',E'',k')$ such that  $|V(G')| \leq 32k^3 + 4k$ and 
 $(G',E'',k')$ is a YES instance if and only if $(G,E',k)$ is a YES instance. 
 Let us remark that  we put no efforts to optimize the size of the kernel. 

  Let $A,B$ be a partitioning of the vertices of graph $G$ in the given instance $(G,E',k)$. Initially we put 
  $A = \emptyset$ and $B = V(G)$. 
   There is a sequence 
of procedures. To avoid a confusing nesting of if-then-else statements, we use the 
following convention: The first case which applies is used first in the procedure. 
Thus, inside a given procedure, the hypotheses of all previous procedures are assumed false.
  
\begin{enumerate}
 \item[P1:] If $4k < |A|$ then return a trivial NO instance, else if $G[B]$ contains a chordless cycle $W$ of length 
at least four, move $V(W)$ from $B$ to $A$.

 \item[P2:] If $4k < |A|$ then return a trivial NO instance, else if $G[B]$ contains a chordless path $P$ of at least two vertices 
which is also an induced subgraph of a chordless cycle $W$ of $G$ of length at least four, move $V(P)$ from $B$ to $A$.

 \item[P3:] Compute $A_{x,y}$ for each pair of vertices in $A$.

 \item[P4:] If $|A_{x,y}| \leq 2k$ then move $A_{x,y}$ from $B$ to $A$, else if $xy \not\in E'$ return a trivial NO instance, else 
add edge $xy$ to $F$.

 \item[P5:] Delete every vertex of $B$.

%

\end{enumerate}
Now we argue on the correctness and running time required to implement the procedures. 
P1: Every cordless cycle of length $4 \leq \ell$ requires at least $\ell -3$ fill edges to be triangulated~\cite{KaplanST99}. 
Thus, at least one fill edge has to be added for each $4$ vertices moved to $A$, and 
we have a NO instance if $4k < |A|$.
A chordless cycle of a non chordal graph can be obtained in $\cO(n+m)$ time~\cite{TarjanY84}. 
Total running time is $\cO(k(n+m))$. 

P2: Let $P$ be an induced path which is an induced subgraph of a chordless cycle $W$ in a graph $G$. 
Then any triangulation of $G$ will add at least $|V(P)|-1$ fill edges incident to vertices in $P$~\cite{KaplanST99}.
Thus, $2$ end points of a fill edge(equivalently to one fill edge) has to be added for each $4$ vertices moved to $A$, and 
we have a NO instance if $4k < |A|$.
A path $P$ can be obtained as follows: Let $u$ be an end vertex of $P$ and let $x$ be the unique neighbour of $u$ in $A$ 
on the cycle $W$. We have a chordless cycle $W$ satisfying the conditions if there is a path from $(N(u) \cap B) \sm N(x)$ 
to a vertex in $N(x) \sm N(u)$ not using vertices of $N(x) \cap N(u)$. Such a path can trivially be found in $\cO(m)$ time. 
Thus the time required to this step is $\cO(knm)$. 

P3: Let $x,y$ be nonadjacent vertices of $A$. For each vertex $w \in N(x) \cap N(y) \cap B$ check if there is a path from 
$x$ to $y$ avoiding $N[w]$, if so there is a chordless cycle containing $xwy$ as consecutive vertices. 
Running time is $\cO(k^2nm)$  because  $|A|$ is $O(k)$.

P4: If $|A_{x,y}| > 2k$ then by the Safe-edge Rule, it is safe to add edge $xy$ if $xy \in E'$, and to 
return a trivial NO instance if  $xy \not\in E'$.
Running time for this step is $\cO(k^2n)$. 

P5: For every remaining vertex $u \in B$, No-cycle-vertex Rule can be applied, and vertex $u$ can safely be 
deleted from the graph $M = (V,E \cup F)$. 
Let us on the contrary assume that this is not the case. 
Then one can find   two nonadjacent vertices $a,b \in N_M(u)$ and 
 connected component $C$ of $M \sm N_M[w]$ such that $a,b \in N_M(C)$. 
Let $W$ be the chordless cycle of length at least four, obtained by $yux$ 
and a shortest chordless path from $x$ to $y$ in $M[C \cup \{x,y\}]$.
By P1, at least one vertex of $W$ is contained in $A$ and 
by P2, no two consecutive vertices of $W$ are contained in $B$. 
By our assumption $u \in B$, so $x,y \in A$ by P2.
Vertex $u$ is contained in $N(x) \cap N(y)$ because during  P3 we add  edges only between vertices in $A$. 
Furthermore, $u \in A_{x,y}$ because otherwise by \cite[Theorem 2.10]{KaplanST99}, $u$ would already have been moved to $A$ by P2.
 Since $xy \not \in F$ and $u \in A_{x,y} \cap B$,   we have a contradiction to P4.  

Let $G' = (A, E(G[A] \cup F)$ and let $E''$ be the edges in $E' \sm F$ where both endpoints are incident to vertices of $A$.
 Since $(G', E'', k - |F|)$ is obtained 
by applying   Safe-edge Rule on each edge in $F$ and  No-cycle-vertex Rule on each vertex in $B$, we have that instance $(G', E'', k - |F|)$ is   a YES instance if and only if 
$(G,E',k)$ is a YES instance,
Finally, $|A| \leq 32k^3 + 4k$ because during P1 and P2 we move in total at most $4k$ vertices from $B$ to $A$, 
and during  P3  at most $2k \cdot (4k)^2$ vertices to $A$.

\end{proof}

 \begin{theorem}\label{thm:chordsand}  \textsc{Chordal Sandwich Problem}  is solvable in time 
$\cO(2^{\cO(\sqrt{k}\log{k})} +k^2 nm)$.
\end{theorem}

\begin{proof} Let   $(G_1 , G_2 ,k)$ be an instance of the problem. We sketch the proof by following the steps of the proof of Theorem~\ref{thm:mainthm} and commenting on the differences. 
 \emph{Step A:} We use Lemma~\ref{lem:kernel_sandwich}, to obtain   in time $\cO(k^2 nm)$ kernel   $(G'_1 , G'_2 ,k')$ such that   $|V(G'_1)|= \cO(k^3 )$ and $k'\leq k$.
 \emph{ Step B1:}  On kernel we use Branching Rule~\ref{br:rule1} exhaustively, with the adaptation that every instance defined by fill edge set 
$F_i$ where $F_i \not\subseteq E_2$ is  discarded. Thus we obtain  $2^{\cO(\sqrt{k}\log{k})} $ non-reducible instances.
  \emph{Step B2:} For each non-reducible instance $(G^i_1, G^i_2,k_i)$,  we enumerate vital  potential maximal cliques of $G_1^i$  but discard all potential maximal cliques that are not cliques 
in $G^i_2$. 
  \emph{Step C:} Solve the remaining problem in time proportional to the number of vital  potential maximal cliques in $G_1'$  that are also cliques in $G_2^i$. This step is almost identical to Step~C of Theorem~\ref{thm:mainthm}.  
\end{proof}
 
 \noindent\textbf{Triangulating Colored Graph.}
In the {\sc Triangulating Colored Graph} problem we are given a graph $G=(V,E)$ with a partitioning of $V$ into sets 
$V_1,V_2,\ldots,V_c$, a colouring of the vertices. Let us remark that this coloring is not necessarily a proper coloring of $G$. 
The question is if $G$ can be triangulated without adding  edges between vertices in the same set (colour). 
\begin{center}
\fbox{\begin{minipage}{13cm}
\noindent {\sc Triangulating Colored Graph}\\
{\sl Input:} A graph $G = (V,E)$, a partitioning of $V$ into sets $V_1,V_2,\ldots,V_c$, and an integer $k$.\\
{\sl Question:} Is there $F \subseteq [V]^2$, $|F|\leq k$, and such that for each $uv \in F$, $u,v \not \in V_i$, $1\leq i\leq c$, and  graph $H=(V,E_1 \cup F)$ is a triangulation of $G$? 
\end{minipage}}
\end{center}

{\sc Triangulating Colored Graph} 
 can be trivially reduced to   \textsc{Chordal Graph Sandwich}   by defining $G_1=G$, and 
 the edge set of graph $G_2$ as the edge set  of $G_1$ plus the set of  all vertex pairs of different colors.  
 Thus by Theorem~\ref{thm:chordsand},  {\sc Triangulating Colored Graph} is solvable in time $\cO(2^{\cO(\sqrt{k}\log{k})} +k^2 nm)$.
 
%

%



\section{Conclusions and open problems}\label{sec:conclusion}
In this paper we gave the first parameterized subexponential time algorithm solving  \textsc{Minimum Fill-in} in time $\cO(2^{\cO(\sqrt{k}\log{k})} +k^2nm)$. It would be interesting to find out how tight is the  exponential dependence, up to some complexity assumption, in the running time of our algorithm.   
We would be    surprised  to hear about  time $2^{o(\sqrt{k})}n^{\cO(1)}$ algorithm solving  \textsc{Minimum Fill-in}. For example, such an algorithm would be able to solve the problem in time $2^{o(n)}$. However, the only results we are aware in this direction is that  \textsc{Minimum Fill-in} cannot be solved in time $2^{o(k^{1/6})}n^{\cO(1)}$ unless the ETH fails \cite{Cyganprivate}. See 
\cite{ImpagliazzoPZ01} for more information on ETH.  Similar uncertainty  occurs with a number of
other graph problems expressible in terms of vertex orderings. Is it possible to prove that unless ETH fails, there are no $2^{o(n)} $ algorithms for \textsc{Treewidth}, \textsc{Minimum Interval Completion}, and  
 \textsc{Optimum Linear Arrangement}? Here the big gap between what we suspect   and what we know is frustrating. 
 
 On the other hand, for the \textsc{Triangulated Colored Graph}  problem,  which we  are able to solve in time $\cO(2^{\cO(\sqrt{k}\log{k})} +k^2nm)$, Bodlaender et al.  \cite{BodlaenderFW92} gave  a polynomial time reduction  that from a 
 $3$-SAT formula on $p$ variables and $q$ clauses constructs an instance of  
\textsc{Triangulated Colored Graph}. This instance has $2+2p+6q$ vertices and  a triangulation of the instance respecting its
colouring can be obtained by   adding of at most $(p+ 3q) + (p+3q)^2 + 3pq$ edges.  Thus up to  ETH, 
\textsc{Triangulated Colored Graph}  and  {\sc Chordal Graph Sandwich}  cannot be solved in time $2^{o(\sqrt{k})}n^{\cO(1)}$.

The possibility of improving the $nm$ factor in the  running time $\cO(2^{\cO(\sqrt{k}\log{k})} +k^2nm)$ of the algorithm is another interesting open question. The factor  $nm$ appears from the running time required by the kernelization algorithm to identify simplicial vertices. Identification of simplicial vertices can  be done in time 
$\cO(\min\{mn,n^\omega \log{n}\})$, where $\omega < 2.37$ is the exponent of matrix multiplication  \cite{HeggernesTV05,KratschS06}. 
   Is  the running time required to obtain a polynomial kernel for \textsc{Minimum Fill-in} is at least the time required to identify a simplicial vertex in a graph and can search of  a simplicial vertex be done  faster than finding a triangle in a graph?

Finally, there are various  problems in graph algorithms, where the task is to find a  minimum number of edges or vertices to be changed such that the resulting graph belongs to  some graph class.   For example, the problems of completion to interval and proper interval graphs are  fixed parameter tractable  \cite{HeggernesPTV07,focs/KaplanST94,KaplanST99,Villanger:2009ez}. Can these problems be solved by subexponential parameterized algorithms? Are there any generic arguments explaining why some FPT graph modification problems can be solved in subexponential time and some can't? 
%
 
\medskip
\noindent
\textbf{Acknowledgement}  We are grateful to Marek Cygan and Saket Saurabh for discussions and useful suggestions.

\begin{small}

\end{small}

%
%
%
%
%
%
%
%

\end{document}